\documentclass[aps,amssymb,amsmath,amsfonts,superscriptaddress]{revtex4}

\usepackage{graphicx}
\usepackage{epstopdf}
\usepackage{amsthm}
\usepackage{amsmath}
\usepackage{color}

\usepackage{hyperref}
\usepackage{graphicx}
\usepackage{graphics}
\usepackage[T1]{fontenc}
\usepackage[english]{babel}
\usepackage[latin2]{inputenc}

\usepackage{subfig}

\newtheorem{proposition}{Proposition}
\newtheorem{lemma}{Lemma}

\newtheorem{theorem}{Theorem}

\newtheorem{corollary}[proposition]{Corollary}

\newcommand{\scalar}[2]{\langle #1 | #2 \rangle}
\newcommand{\ketbra}[2]{| #1 \rangle \langle #2 |}
\newcommand{\ket}[1]{| #1 \rangle}
\newcommand{\bra}[1]{\langle #1 |}
\newcommand{\tr}{\mathrm{Tr}}

\begin{document}
\title{%(Direct sum) 
    Majorization %entropic 
uncertainty relations
  for mixed quantum states}

\author{Zbigniew Pucha{\l}a}
\email{z.puchala@iitis.pl}
\affiliation{Institute of Theoretical and Applied Informatics, Polish Academy
of Sciences, ul. Ba{\l}tycka 5, 44-100 Gliwice, Poland}
\affiliation{Faculty of Physics, Astronomy and Applied Computer Science, 
Jagiellonian University, ul. {\L}ojasiewicza 11,  30-348 Krak{\'o}w, Poland}

\author{ \L ukasz Rudnicki} 
\affiliation{Max-Planck Institute for the Science of Light, Staudtstra{\ss}e 2, 91058 Erlangen, Germany}
\affiliation{Center for Theoretical Physics, Polish Academy of Sciences, 
Al. Lotnik{\'o}w 32/46, 02-668 Warsaw, Poland}

\author{Aleksandra Krawiec}
\affiliation{Institute of Theoretical and Applied Informatics, Polish Academy
of Sciences, ul. Ba{\l}tycka 5, 44-100 Gliwice, Poland}
\affiliation{Institute of Mathematics, University of Silesia,
ul. Bankowa 14, 40-007 Katowice, Poland}

\author{Karol {\.Z}yczkowski}
\affiliation{Faculty of Physics, Astronomy and Applied Computer Science, 
Jagiellonian University, ul. {\L}ojasiewicza 11,  30-348 Krak{\'o}w, Poland}
\affiliation{Center for Theoretical Physics, Polish Academy of Sciences, 
Al. Lotnik{\'o}w 32/46, 02-668 Warsaw, Poland}

\date{September 29, 2017}

\begin{abstract}
Majorization uncertainty relations are generalized for an arbitrary
mixed quantum state $\rho$ of a finite size $N$. In particular, a lower bound for the sum of
two entropies characterizing probability distributions corresponding to
measurements with respect to arbitrary two orthogonal bases is derived in terms
of the spectrum of $\rho$ and the entries of a unitary matrix $U$ relating both
bases. The obtained results can also be formulated for two measurements performed on
a single subsystem of a bipartite system described by a pure state, and consequently expressed
as uncertainty relation for the sum of conditional entropies.
\end{abstract}


%\pacs{03.65.Aa}
\maketitle

%%%%%%%%%%%%%%%%%%%%%%%%%%%%%%%%%%%%%%%%%%%%%%%%%%%%%%%%%%%%%%%%%%%%%%%%%%%%%%%%
\section{Introduction}
%%%%%%%%%%%%%%%%%%%%%%%%%%%%%%%%%%%%%%%%%%%%%%%%%%%%%%%%%%%%%%%%%%%%%%%%%%%%%%%%
Even though uncertainty relations (URs) are often considered to be a hallmark of quantum mechanics, the field devoted to quantification of fundamental, quantum mechanical uncertainty recently undergoes an accelerating progress (cf. \cite{BR11,CBTW17} and references therein). Quite understandably, current research efforts are not much oriented towards the well-established original formulations due to Heisenberg \cite{He27}, Kennard
\cite{Ke27} and Robertson \cite{Ro29}, involving the product of two variances,
%of expectation values of two quantum measurements. A later approach related to
%the theory of information involved the entropy of probability distribution
%characterizing each quantum orthogonal measurement. 
or their pioneering entropic counterpart by
Bia{\l}ynicki-Birula and Mycielski \cite{BBM75}.
Rare exceptions (few examples rather than a comprehensive list) from that ''rule'', such as geometrical description of multimode position-momentum uncertainty relations \cite{Weigert} by means of smooth functions of second moments, a scheme for entanglement detection based on Robertson UR for three mutually unbiased directions \cite{3MUBSCV}, or covariance-dependent improvements of \cite{BBM75} driven by phase-space description of (non)Gaussian states in quantum optics \cite{Hertz}, share a common conceptual root. It is recognized that the URs in general are not only an inherent part of the formalism. Instead, they prove to be useful and efficient in providing relevant pieces of information in various practical cases,  including experimental ones. 

%a lower bound for the sum of two
%entropies corresponding to two measurements was derived for the case of a
%quantum state belonging to an infinite Hilbert space. 

A popular way of describing uncertainty for states belonging to an $N$--dimensional Hilbert space in terms of information entropies is based on famous results by Deutsch \cite{De83} and Maassen-Uffink \cite{MU88}. The bounds derived in \cite{De83,MU88} valid for the sum of both entropies were expressed as a function of the modulus of the
largest entry of the unitary matrix $U$, which relates both measurement bases.
Uncertainty relations optimized first for two-dimensional Hilbert space
\cite{SR98,GMR03}, were later studied in higher dimensions \cite{VS08,WW10} as well as strengthened, recast \cite{Ra12,CCYZ12,KTW14} and extended to: more than two orthogonal measurements
\cite{WW10} and generalized quantum measurements
\cite{Bosyk3,RZ16}; presence of quantum memory \cite{BCRR10,CF15} and quantum
correlations \cite{CP14}; generalized entropies  \cite{BB06,TR11,ZBP14}, etc.

The sole entropic formalism fits well into thermodynamic considerations on a 
quantum level~\cite{Urbanik}. However, the plethora of associated uncertainty 
relations mentioned above has not played a visible role so far in the growing 
field of quantum thermodynamics. It might be considered surprising since, for 
instance, the generalized second laws of quantum thermodynamics 
\cite{secondlaws} involve the whole family of generalized R{\'enyi} information 
entropies. Similarly, a related notion of quantum coherence \cite{Plenio} with 
its already broad scope ranging from thermodynamics \cite{Aberg}, via quantum 
transport \cite{Abu}, to interference experiments \cite{Janos,Kai}, was 
scarcely (with a single exception \cite{PRCPZ15}) studied with the help of the 
rich machinery of URs. The reasons behind that seem twofold. First of all, in 
applications in which one is interested in keeping track of certain processes, 
one often focuses on interrelations between sole probabilities, rather that 
between the combined measures such as entropies. Moreover, in most of the cases 
it is important to take into account that the quantum states in question are 
mixed, while the standard URs deal with pure states. Of course, most of the 
known URs do remain valid for all mixed states, though they usually contain no 
valuable improvements, as compared to the bare case of pure states.

Possible remedy for the first issue is brought by the concept of majorization.
Making use of the majorization relations between the probability vectors, it was
possible to improve uncertainty relations available for pure states and
in fact formulate them for any Schur concave function \cite{FGG13,PRZ13}. On a 
more general level, the formal theory of uncertainty quantification based on 
majorization has been established \cite{Gour16} and is promised to 
substantially broaden a class of available operational interpretations, likely 
in the spirit of resource theories. In regard to the second limitation 
mentioned, some attempts to extend (entropic)
uncertainty relations to the case of mixed states were reported in
\cite{CP14,RPZ14,KLJR14,kurzyk2017conditional}. Though, an effort to  combine the majorization techniques with the full description of a quantum state in terms of its density matrix has not yet been taken.

Even though recent literature on entropic uncertainty relations grows fast \cite{BR11,CBTW17}, several
interesting problems do remain open. Thus, the technical aim of this work is to supplement the direct-sum
majorization uncertainty relations derived in \cite{RPZ14} by non-trivial information about mixedness of the quantum state.
%and to compare results obtained with other bounds in the literature. 
We believe the presented results enrich the toolbox of the quantum information theory relevant for studies on quantum thermodynamics, wave-particle duality and beyond. Following such a goal, we also consider a bipartite pure state subject to two
orthogonal measurements performed locally on a single subsystem. As a by-product,
we derive lower bounds for the sum of conditional entropies, which can be
expressed in terms of the mutual information.

%As a side remark, let us mention that apart from looking for a lower bound for the sum of two entropies one also 
%studied the upper bound for this sum \cite{Sa93} and showed that the difference of both
%bounds is minimal if all measurement bases are mutually unbiased \cite{PRCPZ15}.
%Entropic uncertainty and certainty relations in the asymptotic case, $N\to
%\infty$, were analysed in \cite{ALPZ16} for a random choice of unitary matrices
%determining the measurement bases.

The present paper is organized as follows. In Section~2 we review necessary
notions and introduce notation used in the paper. Main majorization relation and
the corresponding entropic uncertainty relations are formulated in Section~3.
Comparison of the strength of the obtained bounds with other results from the
literature is presented in Section~4. The paper is concluded in Section~5, while
a proof of the lemma used to prove the main theorem is postponed to Appendix.

\section{Preliminaries} 
We shall start by fixing the basic notation.
For a density matrix $\rho$ of a finite dimension $N$ one defines its von 
Neumann entropy $S$  as
\begin{equation}
S(\rho) = - \tr \rho \ln \rho = - \sum_{j=1}^N \lambda_j \ln \lambda_j,
\end{equation}
where the vector $(\lambda_1, \dots , \lambda_N)$ denotes the spectrum of $\rho$.
A von Neumann measurement given by  an orthogonal basis $\{\ket{a_i}\}_i$ 
yields the probability vector, $p_i = \bra{a_i} \rho \ket{a_i}$, $i=1,\dots, N$
characterized  by the Shannon entropy $H(p)=- \sum_{i=1}^N p_i \ln p_i$ of the measurement outcomes. We can see that, on the one hand the von Neumann entropy of $\rho$  is equal to the Shannon entropy of the vector $\lambda$, while on the other hand the Shannon entropy of the measurement probabilities $p$ is equal to the  von Neumann entropy of the post-measurement state (its eigenvalues are $p_i$).

Thus, for a bipartite state $\rho_{AB}$ acting on a composite Hilbert space ${\cal
H}_A \otimes {\cal H}_B$ one defines the joint entropy
\begin{equation}
%\begin{split}
 S(\rho_{AB}) \equiv H(AB),
\end{equation} 
where $ H(AB)$ is the Shannon entropy of the spectrum of $\rho_{AB}$. By 
introducing the argument '$AB$' instead of the probability vector, we 
modify the notation a bit to make a connection with standard information 
theory. We believe no misunderstanding shall occur due to that step. The 
reduced density matrices $\rho_A = \tr_B \rho_{AB}$ and
$\rho_B = \tr_A \rho_{AB}$
allow one to introduce the entropies of both reductions, 
\begin{equation}
H(A) = S(\rho_A), % = H(\tr_B \rho), \ \
\ \ \ 
H(B) =S(\rho_B).%  = H(\tr_A \rho).
\end{equation} 
In Section~4 we are also going to use the conditional von Neumann entropy,
\begin{equation}
\label{condentr}
H(A|B) = H(AB) - H(B),
\end{equation}
and the mutual information,
\begin{equation}
\label{mutual}
I(A;B) = H(A)+H(B)-H(AB) \equiv I(B;A).
\end{equation}

\medskip

As the majorization uncertainty relations led to lower bounds for the sum of entropies, we shall also introduce two generalized variants of the Shannon entropy, namely,  due to R{\'e}nyi
\begin{equation}
H_{\alpha}(p) = \frac{1}{1-\alpha} \ln \left(\sum_i p_i^{\alpha}\right),
\end{equation}
and Havrda--Charv\'{a}t, popularized in physics community by Tsallis
\begin{equation}
T_{\alpha}(p) = \frac{1}{1-\alpha} \left(\sum_i p_i^{\alpha} -1 \right).
\end{equation}
It is easy to see that the first quantity is additive while
 the second one is not. In the case $\alpha \to 1$, both quantities
reduce to the Shannon entropy $H(p)$.

Consider two real vectors $x$ and $y$ of sizes $N_1$ and $N_2$, respectively. We
associate vectors $x^{\downarrow}, y^{\downarrow}$ of size $\max \{N_1, N_2\}$
with coefficients ordered decreasingly and zeros on additional coordinates added
to the shorter vector. The vector $x$ is said to be majorized by $y$, written $x
\prec y$, if $x^{\downarrow}$ and $y^{\downarrow}$ satisfy inequalities for all
partial sums and sum up to a common value
\begin{equation}
\begin{split}
\sum_{i=1}^m x_i^{\downarrow} &\leq \sum_{i=1}^m y_i^{\downarrow}
\ \ \text{ for } m \leq \max(N_1,N_2),\\
\sum_{i=1}^{N_1} x_i &=\sum_{i=1}^{N_2} y_i.
\end{split}
\end{equation}
Majorization relation between two vectors allows one to write down inequalities
between any Schur--concave functions of both arguments \cite{MO79}. This
technique applied to generalized entropies was used in \cite{PRZ13,FGG13} to
derive  majorization entropic uncertainty relations for two orthogonal
measurements performed on a single system described by a pure state. In the
present work even stronger bounds derived in \cite{RPZ14} will be extended for
arbitrary mixed quantum states.

%%%%%%%%%%%%%%%%%%%%%%%%%%%%%%%%%%%%%%%%%%%%%%%%%%%%%%%%%%%%%%%%%%%%%%%%%%%%%%%%%
\subsection{Direct-sum majorization relation}

%%%%%%%%%%%%%%%%%%%%%%%%%%%%%%%%%%%%%%%%%%%%%%%%%%%%%%%%%%%%%%%%%%%%%%%%%%%%%%%%%

While looking for entropic uncertainty relations one attempts to  bound the sum of entropies for two
probability vectors $p$ and $q$, which characterize outcomes of two measurements
of a given quantum state $\rho$. To fix the notation we assume that $p$ and $q$ are probabilities associated to
orthogonal measurements described by two bases mutually related by a unitary matrix $U$,
\begin{equation}
\begin{split}
p_i &= \bra{i} \rho \ket{i}, \\
q_i &= \bra{i} U^\dagger\rho U\ket{i}.
\end{split}
\end{equation}

As in ~\cite{RPZ14}, given a unitary matrix $U$ of size $N$, we define a
set of rectangular submatrices with a fixed perimeter,
\begin{eqnarray}
\mathcal{SUB}(U,k)&=&\{M:\#\textrm{cols}(M)+\#\textrm{rows}(M)=k+1\nonumber\\
&&\text{ and }M\text{ is a submatrix of }U\}.
\end{eqnarray}
The symbols $\#\textrm{cols}\left(\cdot\right)$ and 
$\#\textrm{rows}\left(\cdot\right)$
denote the number of columns and the number of rows respectively. We further define 
coefficients  $s_k$~\cite{PRZ13} 
\begin{equation}\label{def:s_k}
s_{k}=\max\left\{ \|M\|:\; M\in\mathcal{SUB}(U,k)\right\},
\end{equation}
with $\|M\|$ being the operator norm equal to the maximal singular value of $M$.
By construction, $c=s_1 \leq s_2 \leq s_3 \leq \dots \leq s_N =1$.
Next we define a vector $W$ of length $N$,
\begin{equation}
\label{def:W}
W=\left[s_{1},s_{2}-s_{1},\dots ,s_{N}-s_{N-1}\right]^T.
\end{equation}
The direct-sum
majorization uncertainty relation derived in  \cite{RPZ14} tells us that given a pure input state $\rho=\ketbra{\psi}{\psi}$, we have
\begin{equation}
p \oplus q \prec \{1\} \oplus  W.
\end{equation}
\medskip 

\section{Main result}
Inspired by the approach briefly sketched above, we now
consider generalized majorization relations based on the direct sum of two vectors,
\begin{equation}
p \oplus q \prec W^{(\lambda)},
\end{equation}
for some vector $W^{(\lambda)}$ which now shall also depend on the spectrum $(\lambda)$ of a
measured state.  In order to construct the vector $W^{(\lambda)}$ we first introduce an auxiliary real vector $\boldsymbol{S}$,
\begin{equation}
\label{def:S-mixed}
\begin{split}
S_k &= 
\lambda_1 (1 + s_{k-1}) + \dots +\lambda_n (1 + s_{n}) +
\lambda_{n+1} (1 - s_{n}) +
\dots +
\lambda_{k} (1 - s_{k-1}) \ \ \text{ for  } k = 2n,  \\
S_k &= 
\lambda_1 (1 + s_{k-1}) + \dots +\lambda_n (1 + s_{n+1}) +
\lambda_{n+1} + 
\lambda_{n+2} (1 - s_{n+1}) +
\dots +
\lambda_{k} (1 - s_{k-1}) \ \ \text{ for } k = 2 n+1.
\end{split}
\end{equation}
Even though we use here the symbol reserved for the von Neumann entropy, no confusion shall occur since now $\boldsymbol{S}$ is the bold-faced vector. We use the notation $S_k$ in order to make a conceptual connection with $s_k$ defined before for the case of pure states. First elements of the vector $\boldsymbol{S}$ read: $S_1 = \lambda_1$, \ 
$S_2= \lambda_1(1+s_1)+\lambda_2(1-s_1)$ and
$S_3=\lambda_1(1-s_2)+\lambda_2 +\lambda_3(1-s_2)$.
Finally, we define a vector $W^{(\lambda)}$ of length $2 N$
with elements given by the difference,
\begin{equation}
\label{def:Wlambda}
W^{(\lambda)}_{k} = S_k - S_{k-1}
\end{equation}

By definition we can recast the coefficients of $W^{(\lambda)}$ as 
\begin{equation}
\begin{split}
W^{(\lambda)}_{k} &= S_{k} - S_{k-1} =
\lambda_1 (s_{k-1} -s_{k-2}) +
\lambda_2 (s_{k-2} -s_{k-3}) +
\dots + \\
&+
\lambda_{\lceil\frac{k}{2}\rceil-1 }(s_{\lfloor\frac{k}{2} \rfloor +1} - 
s_{\lfloor\frac{k}{2} \rfloor} )+
\lambda_{\lceil\frac{k}{2}\rceil}(s_{\lfloor\frac{k}{2}\rfloor})
+
\lambda_{k}(1-s_{k-1}),
\end{split}
\end{equation}
which can be rewritten with  use of the original vector $W$ defined in 
Eq.~\eqref{def:W}
\begin{equation}
\begin{split}
W^{(\lambda)}_{k} &=
\lambda_1 W_{k-1} +
\lambda_2 W_{k-2} +
\dots +
\lambda_{\lceil\frac{k}{2}\rceil-1 } W_{\lfloor\frac{k}{2} \rfloor +1} 
\\
&+
\lambda_{\lceil\frac{k}{2}\rceil}
\left(\sum_{j=1}^{\lfloor\frac{k}{2}\rfloor} W_j\right)
+
\lambda_{k}
\left(
\sum_{j=k}^d W_j
\right).
\end{split}
\end{equation}
First elements of the vector $W^{(\lambda)}$ have the following form: 
$W^{(\lambda)}_1 = \lambda_1$ and
$W^{(\lambda)}_2=\lambda_1 s_1 + \lambda_2(1-s_1).$

The above notation allows us to state the main theorem:
\begin{theorem}\label{th:main-th}
For a unitary matrix $U$ of size $N$ and a density matrix $\rho$ with
eigenvalues given by ($\lambda$), the following majorization relation holds
\begin{equation}
\label{major17}
p \oplus q \prec W^{(\lambda)}.
\end{equation}
\end{theorem}

The proof of above theorem utilizes the lemma presented below, but proven in Appendix.
\begin{lemma} \label{lemma:general-inequality}
Let $\lambda$ be a given vector of non-negative eigenvalues summing up to one.
Let $\rho$ be a mixed state with eigenvalues given by $\lambda$.
Let $p_i = \bra{i} \rho \ket{i}$ and $q_i =  \bra{a_i} \rho \ket{a_i}$. We 
assume, that $\{\ket{i}\}_i$ and $\{\ket{a_i}\}_i$ form two orthonormal bases. 
Then the following inequality is true
\begin{equation}
p_1+ \dots + p_m + q_1+\dots + q_n \leq \lambda^{\downarrow} \cdot 
\mu^{\downarrow},
\end{equation}
where $\cdot$ denotes the scalar product. The vector $\mu$ of length $m+n$ is
defined as follows,
\begin{equation}
\mu = \{1,1,\dots, 1\}  + (\sigma(A) \oplus (- \sigma(A)),
\end{equation}
where the matrix $A$ has entries $A_{ij} = \scalar{a_i}{j}$, $\sigma(A)$ denotes
the vector of singular values of  $A$ and the symbol $\oplus$ represents
concatenation of two vectors. If necessary, the vector $\sigma(A) \oplus (-
\sigma(A))$ is extended to the length $m+n$ with elements equal to zero.
\end{lemma}
For instance, in the case of $m=1$ and $n=4$, the matrix $A$ has one nonzero
singular value and
\begin{equation}
\mu^{\downarrow} = \{1+\sigma_1(A),1,1,1,1-\sigma_1(A)\}.
\end{equation}

Next proposition is necessary to  obtain the desired majorizing vector.
\begin{proposition}\label{prop:general-bound}
Let $n \leq m$. 
With notation presented above
 the following  inequalities are satisfied,
\begin{equation}
\begin{split}
&p_1+p_2+ \dots + p_m + q_1 +q_2 + \dots + q_n \leq
\\
&\leq
\lambda_1 (1 + s_{n+m-1}) +
\lambda_2 (1 + s_{n+m-2}) +
%\lambda_3 (1 + s_{n+m-3}) + 
\dots +\lambda_n (1 + s_{m}) +
\\
&+ 
\lambda_{n+1}+
\lambda_{n+2}+
\dots  +
\lambda_{m}+
\\
&+ 
\lambda_{m+1} (1 - s_{m}) +
\lambda_{m+2} (1 - s_{m+1}) +
\dots +
\lambda_{m+n} (1 - s_{m+n-1}) 
\end{split}
\end{equation}
\end{proposition}
\begin{proof}
Using Cauchy interlacing property for singular values~\cite[Cor. 3.13]{hj2}
we obtain 
\begin{equation}
\sigma_k(A) \leq  \sigma_1(A^{(k-1)})
\end{equation}
where $A^{(k)}$ is a submatrix of $A$ obtained by deleting $k$ rows and/or 
columns from $A$. Substituting the above inequalities to 
Lemma~\ref{lemma:general-inequality} we complete the proof of 
the Proposition \ref{prop:general-bound}.
\end{proof}

\medskip

\proof[Proof of Theorem~\ref{th:main-th}]
In order to prove the main theorem we consider the largest possible bound from the 
Proposition~\ref{prop:general-bound},  with restriction that $n+m = k$.
Using elementary inequalities one notes, that the largest possible value is in 
the case when $n = m$ for  even $k=2 n$ and $n=m-1$ for odd $k=2n+1$. In this 
way we recover values $S_k$ defined in Eq.~\eqref{def:S-mixed}, which give us 
the desired majorization relation (\ref{major17}).

\hfill $\square$

Using Theorem~\ref{th:main-th} we are in position to upgrade the  entropic uncertainty 
relations for generalized entropies:
\begin{theorem}
With notation as above the following entropic uncertainty relations hold:
\begin{itemize}
\item[a)] for the Shannon entropy $H$,
\begin{equation}
\label{shann_mixed}
H(p) + H(q) \geq - \sum_{i} W^{(\lambda)}_i \ln W^{(\lambda)}_i;
\end{equation}
\item[b)] for the R\'enyi entropy $H_{\alpha}$ of order $\alpha<1$,
\begin{equation}
H_{\alpha}(p) + H_{\alpha}(q) \geq 
\frac{1}{1-\alpha}  \ln \left(\sum \left(W^{(\lambda)} \right)^{\alpha} - 1 
\right);
\end{equation}
\item[c)] for the Tsallis entropy $T_{\alpha}$ of any order $\alpha \geq 0$,
\begin{equation}
T_{\alpha}(p) + T_{\alpha}(q) \geq \frac{1}{1-\alpha} 
\left(\sum \left(W^{(\lambda)} \right)^{\alpha} - 2 \right).
\end{equation}
\end{itemize}
\end{theorem}

\proof
The above inequalities follow from majorization relation
 (\ref{major17})
and the Schur concavity of generalized entropies. In the case 
of the R\'enyi entropy we use also 
the subadditivity of the function $\ln(1+x)$.

\hfill $\square$

%%%%%%%%%%%%%%%%%%%%%%%%%%%%%%%%%%%%%%%%%%%%%%%%%%%%%%%%%%%%%%%%%%%%%%%%%%%%%%%%
\section{Purification of mixed states and bipartite pure states}
%%%%%%%%%%%%%%%%%%%%%%%%%%%%%%%%%%%%%%%%%%%%%%%%%%%%%%%%%%%%%%%%%%%%%%%%%%%%%%%%
In the second part we apply the mixed-states majorization uncertainty relation derived above to the case of composite systems. To this end, we consider a pure state $\ket{\psi}$ of a bipartite system, namely $\ket{\psi} \in H_A 
\otimes H_B$. 
By performing a measurement described by a basis $\ket{a_i}$ on the subsystem $A$, we obtain a set of
outcomes, labeled by $i$, with their corresponding probabilities  
\begin{equation}
p_i = \tr \ketbra{\psi}{\psi} \left(\ketbra{a_i}{a_i} \otimes I_B \right)
= \tr \left(\tr_B \ketbra{\psi}{\psi}\right) \ketbra{a_i}{a_i} 
=\bra{a_i} \left(\tr_B \ketbra{\psi}{\psi}\right) \ket{a_i}.
\end{equation}
To create a scenario relevant for quantum mechanical uncertainty relations we shall now consider two von Neumann measurements $X$ and $Y$ performed on part $A$. For the sake of simplicity we assume that the first measurement is given by the computational basis, while the 
second one is provided by $\ket{a_i}$. These two measurements transform a bipartite (in general mixed)  input state as follows:
\begin{equation}\label{Trans}
\begin{split}
\rho_{AB} &\mapsto 
\rho_{XB} :=
\sum_{i}
\left(\ketbra{i}{i} \otimes I_B \right)
\rho_{AB}
\left(\ketbra{i}{i} \otimes I_B \right),  \\
\rho_{AB} &\mapsto 
\rho_{YB}:=
\sum_{i}
\left(\ketbra{a_i}{a_i} \otimes I_B \right)
\rho_{AB}
\left(\ketbra{a_i}{a_i} \otimes I_B \right). 
\end{split}
\end{equation}
The conditional entropies of the measurement outcomes and the system $B$ are then given by 
by:
\begin{equation}
H(X|B) = H(XB) - H(B),\qquad H(Y|B) = H(YB) - H(B).
\end{equation}

Now, with the help of (\ref{Trans}) we observe that 
\begin{equation}
\begin{split}
H(XB) &= H(\vec{p}),\quad \text{and}\quad H(YB) = H(\vec{q}),
\end{split}
\end{equation}
where
\begin{equation}
p_i = \tr \rho_{AB}\left(\ketbra{i}{i} \otimes I_B \right)
= \tr \left(\tr_B \rho_{AB} \right) \ketbra{i}{i} 
=\bra{i} \rho_A \ket{i},
\end{equation}
and by analogy $q_i= \bra{a_i} \rho_A \ket{a_i}$. Note in passing that whenever the input state $\rho_{AB}$ is pure, so that the entanglement between $A$ and $B$ is fully characterized by the vector of state's Schmidt coefficients $\vec{\lambda}$, we get a simplified expression for the sum of both 
conditional entropies,
\begin{equation}
H(X|B) + H(Y|B) = H(\vec{p})+H(\vec{q}) - 2H(\vec{\lambda}).
\end{equation}
%
%Our results give us 
%\begin{equation}
%H(X|B) + H(Y|B) \ \geq \ 2I(P),
%\end{equation}
%where $I(P)$ is a mutual information defined in Eq. (\ref{mutual}).
%

\section{Conditional entropic uncertainty relations}
Now we are in position to establish the majorization entropic uncertainty relation for the conditional entropies. 
Theorem~\ref{th:main-th} gives us the majorizing vector $W^{(\lambda)}$, which in turn
provides a bound~\eqref{shann_mixed} for the sum of both measurement entropies.
To achieve the desired goal it is convenient to rewrite the majorizing vector in question with the help of matrix notation. To this end, we first define an auxiliary matrix:
\begin{equation}
\Lambda = 
\begin{bmatrix}
\lambda_1,&\color{red}\lambda_1,&\color{red}\lambda_1,&\color{red}\lambda_1,&\color{red}\lambda_1,&\dots,&\color{red}\lambda_1&\dots
 \\
\lambda_1,&\color{red}\lambda_2,&\color{red}\lambda_2,&\color{red}\lambda_2,&\color{red}\lambda_2,&\dots,&\color{red}\lambda_2&\dots\\
\lambda_2,&\color{red}\lambda_1,&\color{red}\lambda_3,&\color{red}\lambda_3,&\color{red}\lambda_3,&\dots,&\color{red}\lambda_3&\dots\\
\lambda_2,&\color{red}\lambda_2,&\color{red}\lambda_1,&\color{red}\lambda_4,&\color{red}\lambda_4,&\dots,&\color{red}\lambda_4&\dots\\
\lambda_3,&\lambda_3,&\color{red}\lambda_2,&\color{red}\lambda_1,&\color{red}\lambda_5,&\dots,&\color{red}\lambda_5&\dots\\
\lambda_3,&\lambda_3,&\color{red}\lambda_3,&\color{red}\lambda_2,&\color{red}\lambda_1,&\dots,&\color{red}\lambda_6&\dots\\
\lambda_4,&\lambda_4,&\lambda_4,&\color{red}\lambda_3,&\color{red}\lambda_2,&\dots,&\color{red}\lambda_7&\dots\\
\lambda_4,&\lambda_4,&\lambda_4,&\color{red}\lambda_4,&\color{red}\lambda_3,&\dots,&\color{red}\lambda_8&\dots\\
\lambda_5,&\lambda_5,&\lambda_5,&\lambda_5,&\color{red}\lambda_4,&\dots,&\color{red}\lambda_9&\dots\\
\lambda_5,&\lambda_5,&\lambda_5,&\lambda_5,&\color{red}\lambda_5,&\dots,&\color{red}\lambda_{10}&\dots\\
\lambda_6,&\lambda_6,&\lambda_6,&\lambda_6,&\lambda_6,&\dots,&\color{red}\lambda_{11}&\dots\\
\lambda_6,&\lambda_6,&\lambda_6,&\lambda_6,&\lambda_6,&\dots,&\color{red}\lambda_{12}&\dots\\
\vdots,&\vdots,&\vdots,&\vdots,&\vdots,&\ddots,&\vdots&\ddots\\
\end{bmatrix}
\end{equation}
which allows us to relate the vector
$W^{(\lambda)}$ defined by (\ref{def:Wlambda})
and the vector $W$ introduced in (\ref{def:W})
\begin{equation}
W^{(\lambda)} = \Lambda W.
\end{equation}
In the next step we consider a joint probability distribution given by a matrix 
\begin{equation}
P=\frac12 \Lambda \ \rm diag(W),
\end{equation}
where $\rm diag(W)$ is a diagonal matrix with the vector $W$ on its diagonal. 
The marginal distributions of $P$ are given 
by $\frac12 W^{(\lambda)}$ and $W$ respectively. 

One can check that $P$ is a permutation of the matrix
\begin{equation}
%P \sim 
\frac12 
\begin{bmatrix}
\lambda \otimes W & 0 \\
0 & \lambda \otimes W 
\end{bmatrix},
\end{equation}
which has an internal structure of the tensor product. This fact implies the formula for the entropy,
\begin{equation}
H(P) =  H(\lambda) + H(W) + \log 2.
\end{equation}
On the other hand, we have that $\frac12 W^{(\lambda)}$  and $W$ are reductions 
of 
$P$, thus 
\begin{equation}
H\left(\frac12 W^{(\lambda)}\right) + H(W) = I(P) + H(P) = I(P) + H(\lambda) + 
H(W) + \log 2,
\end{equation}
where $I(P)$ is the mutual information (\ref{mutual})
of two variables described by a joint 
probability distribution $P$. 
We simplify the above relation to the form
\begin{equation}
H\left(\frac12 W^{(\lambda)} \right) = I(P) + H(\lambda) + \log 2,
\end{equation}
which can also be rewritten as
\begin{equation}
-\sum W^{(\lambda)}_i \ln W^{(\lambda)}_i
= 2 H\left(\frac12 
W^{(\lambda)}\right) - 2 \log 2 = 2 
I(P)+2 H(\lambda).
\end{equation}
This result allows us to express the sum of both conditional entropies and to find an 
information-theoretical interpretation of the relation~\eqref{shann_mixed}.

\begin{corollary}
In the notation presented above we have the following conditional entropic 
uncertainty relation
\begin{equation}\label{RSM}
H(X|B)+H(Y|B) = H(XB) + H(YB) - 2 H(B) \geq 
-\sum W^{(\lambda)}_i \ln W^{(\lambda)}_i
- 2 H(\lambda) = 2 I(P) \geq 0.
\end{equation}

\end{corollary}

It is important to observe that the left hand side of the above inequality is a
concave function of the quantum state, while the right hand side is a convex and
non-negative function of the parameters $\lambda$. The latter follows from the fact
that one marginal distribution of $P$ does not depend on  $\lambda$.

\subsection{One-qubit measurements}
Let us emphasize, that the major results \eqref{RSM} holds for an arbitrary
dimension $N$. Even though, in order to show all its advantages one needs to analyze the
case $N\geq 3$, we start the discussion by presenting a simple $N=2$ example,
for which our approach dominates the previous results (including the strong 
bound of Ref.~\cite{kurzyk2017conditional} specialized for one qubit 
measurements only)
for certain mixed states of high purity. Consider  a normalized $N=2$
mixed state $\rho$ with spectrum $\{ \lambda_1 , \lambda_2 \}$ and assume
that $\lambda_2=1-\lambda_1\le 1/2 \leq \lambda_1$.
We also assume that two orthogonal measurements are taken in bases $\ket{\chi_i}$ and 
$\ket{\varphi_i} $, $i=1, 2$, such that the probability of obtaining 
$i$-th outcome is given by $p_i = \bra{\chi_i} {\rho} \ket{\chi_i}$ and $q_i = 
\bra{\varphi_i} {\rho} \ket{\varphi_i}$ respectively. 
Alternatively, one can consider a pure state $\ket{\psi_{AB}} \in \mathcal{H}_A \otimes \mathcal{H}_B$, such that its partial trace 
is $\tr_B \ketbra{\psi_{AB}}{\psi_{AB}} = \rho$ and analyze local measurements performed on the subsystem $A$.
The spectrum $\{ \lambda_1 , \lambda_2 \}$ of the state, equal to the Schmidt vector of $\ket{\psi_{AB}}$,
allows us to bound from above the following sum of measurement probabilities.
\begin{equation}
\begin{split}
p_i &\leq \lambda_1 \\
p_i+q_j &\leq \lambda_1 (1 + s_1) + \lambda_2 (1 - s_1)\\
p_1+p_2+q_j &\leq 1+ \lambda_1\\
p_1+p_2+q_1+q_2 &= 2.
\end{split}
\end{equation}
The above inequalities render the form of  $W^{(\lambda)}$, the vector known to majorize the direct sum
\begin{equation}
p \oplus q \prec W^{(\lambda)} = \{\lambda_1, \lambda_1 s_1 + \lambda_2(1 - 
s_1), \lambda_1(1 - s_1) +\lambda_2 s_1 , \lambda_2\}.
\end{equation}
% In section ... we show
Note that the reasoning leading to Eq. (\ref{RSM}) implies the following
inequality concerning the sum of two conditional entropies,
\begin{equation}
H(R|B)+H(S|B) \ \geq  \ 2 I(P).
\end{equation}
In Fig.~\ref{fig:qubit-theta-plot} and Fig.~\ref{fig:qubit-lambda-plot} we compare
the above bound with other bounds   for the sum of
conditional entropies known from the literature. In the qubit case we consider the following lower bounds: ($B_B$) by
Berta et al.~\cite{BCRR10}; ($B_{KLJR}$) by Korzekwa et al.~\cite{KLJR14};
($B_{KPP}$) --- a strong bound valid only for single qubit measurements --- by Kurzyk et
al.~\cite{kurzyk2017conditional}; the bound ($B_{RPZ3}$) from~\cite{RPZ14}. It is
important to stress, that in the bounds for conditional entropy we reject the
negative part of the bound.

\begin{figure}[ht]
\centering
\subfloat[$\lambda=\frac13$]{\includegraphics[width=0.5\linewidth]{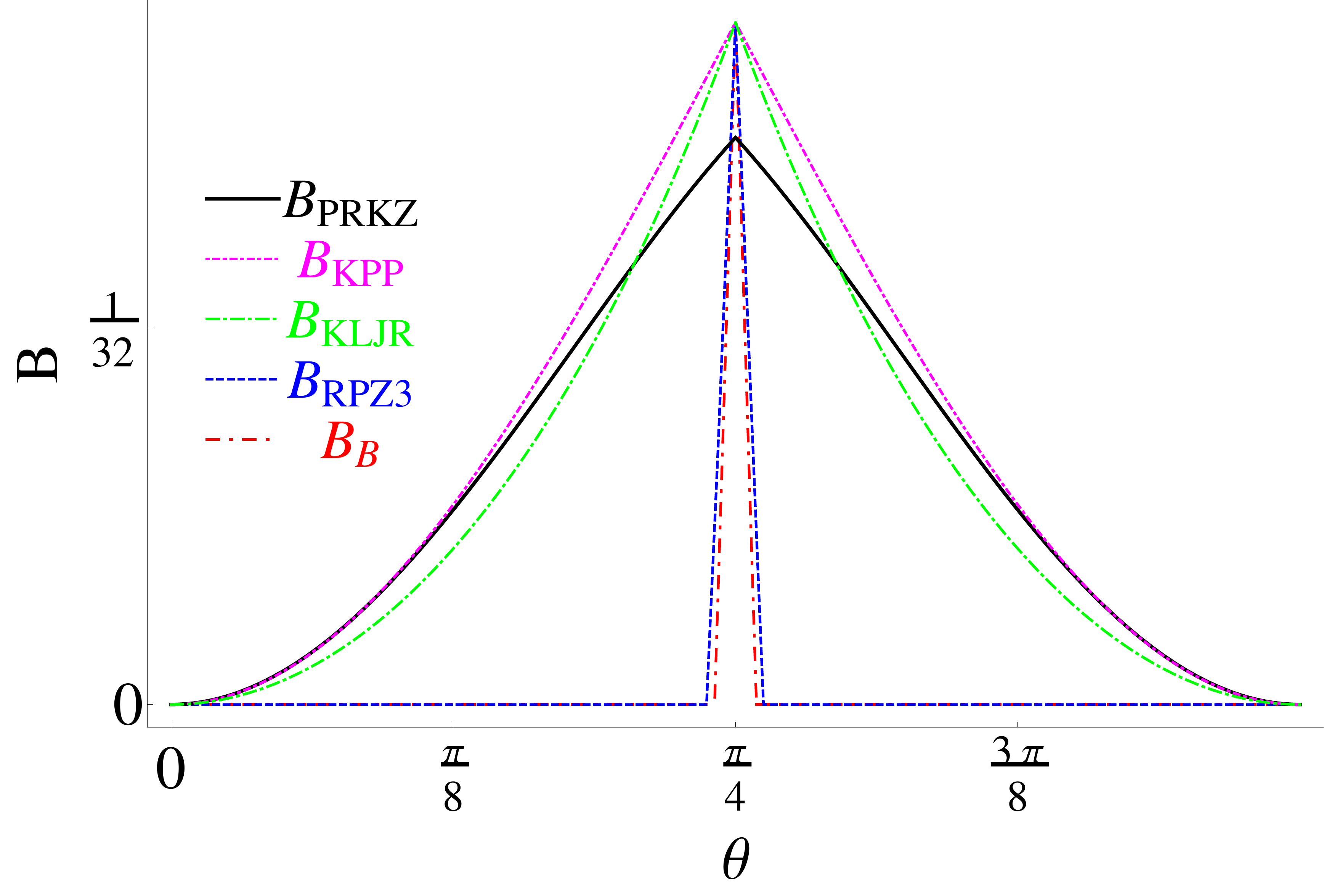}}
\subfloat[$\lambda=\frac16$]{\includegraphics[width=0.5\linewidth]{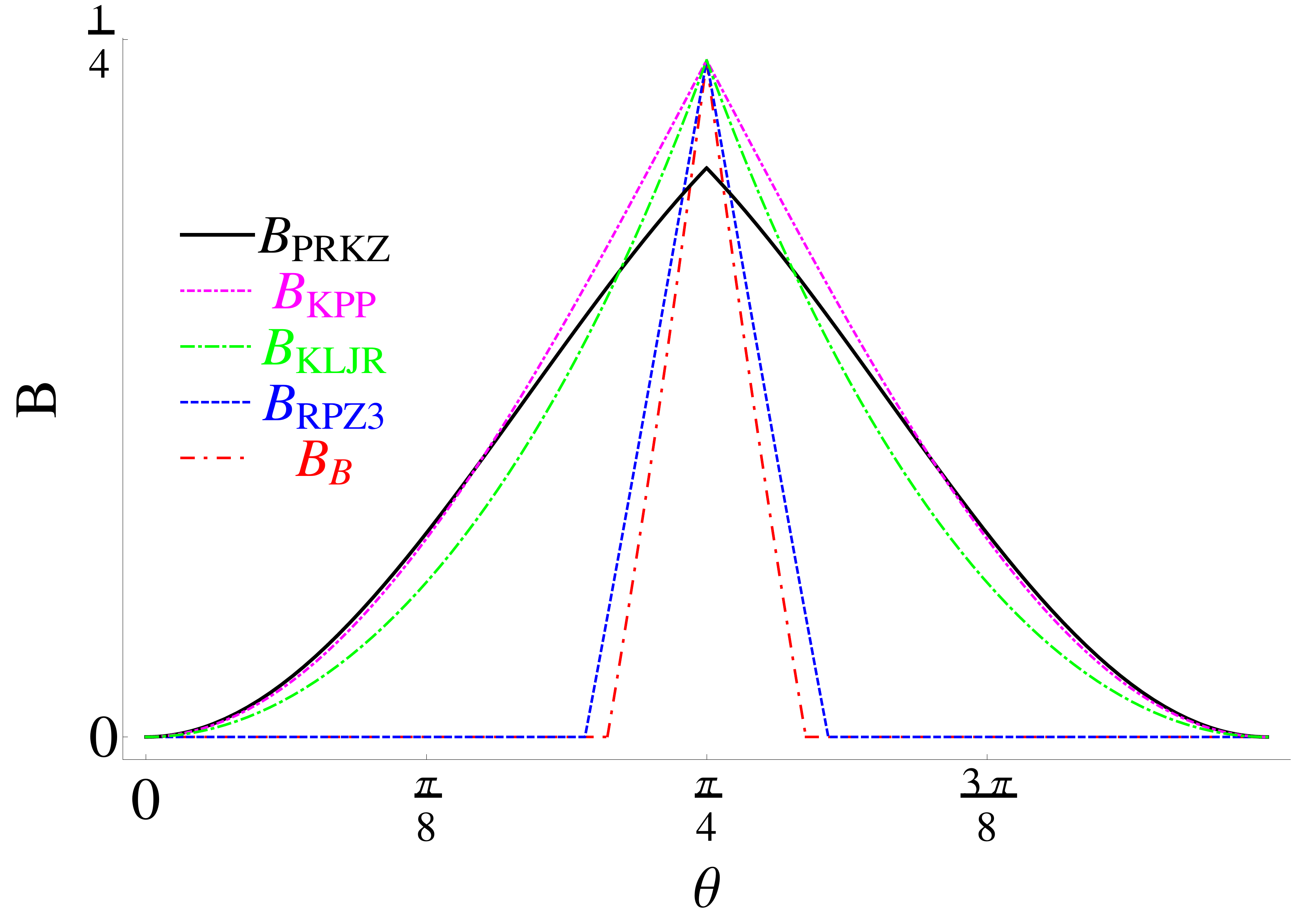}}
\caption{Lower bounds for a sum of conditional entropies in the case of two von
Neumann measurements performed on a single qubit. One measurement is in a computational basis, the
second one is in a basis rotated by an angle $\theta$, obtained for a mixed state
with spectrum $\{\lambda,1-\lambda\}$. Our bound ( {\bf --} thick solid curve 
$B_{PRKZ})$~\eqref{RSM} 
is compared with bound ($B_B$) by Berta et al.~\cite{BCRR10}; ($B_{KLJR}$) by 
Korzekwa et al.~\cite{KLJR14}; ($B_{KPP}$) by Kurzyk et 
al.~\cite{kurzyk2017conditional} (valid for one qubit measurements only) and bound ($B_{RPZ3}$) from~\cite{RPZ14}. } 
\label{fig:qubit-theta-plot}
\end{figure}

\bigskip

\begin{figure}[ht]
\centering
\subfloat[$\theta=\frac{\pi}{3}$]{\includegraphics[width=0.5\linewidth]{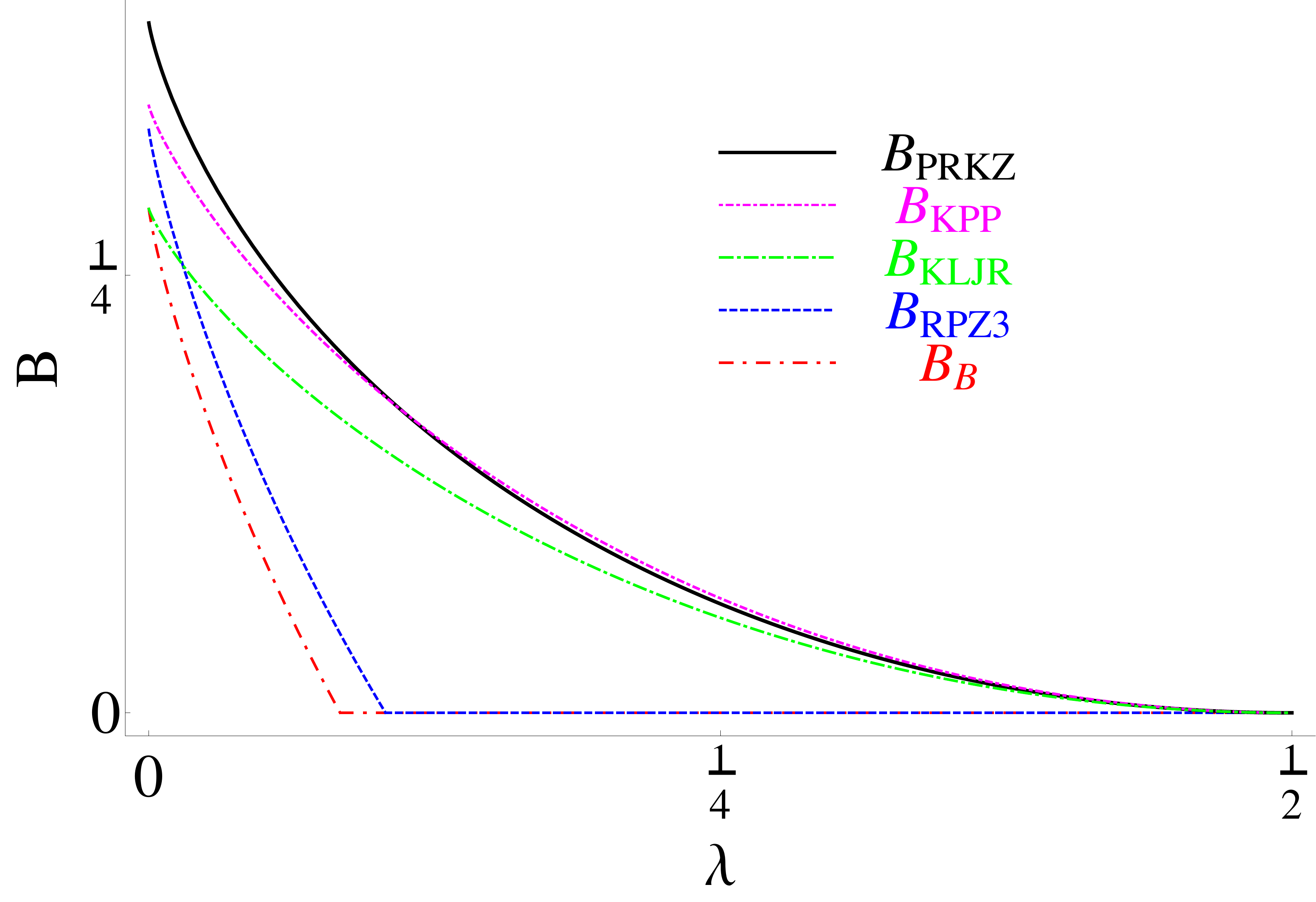}}
\subfloat[$\theta=\frac{\pi}{8}$]{\includegraphics[width=0.5\linewidth]{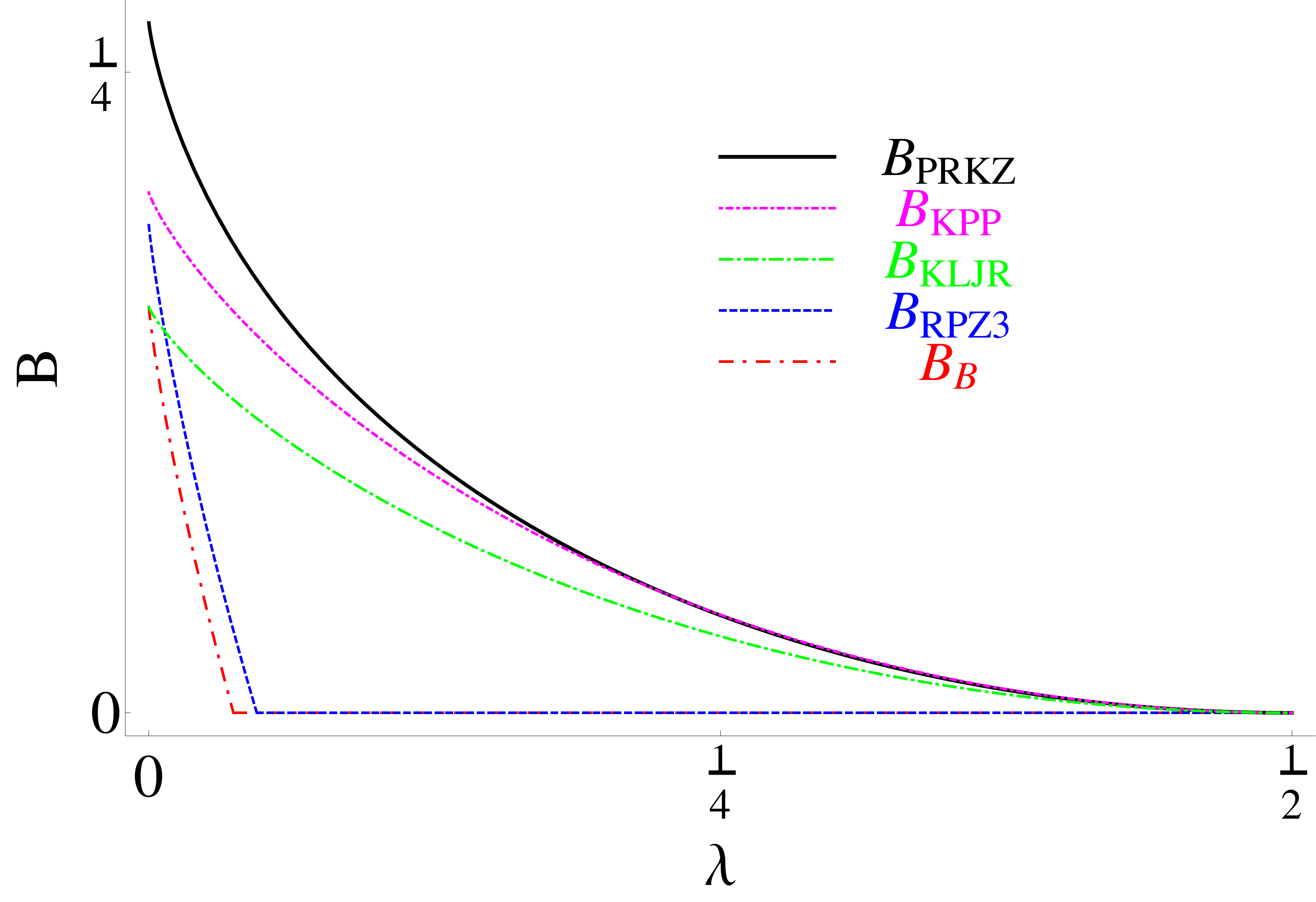}}
\caption{
Uncertainty relations for a one qubit system measured in a basis rotated by an
angle $\theta$, as a function of the smaller eigenvalue $\lambda \in [0,1/2]$.
For mixed states of a high purity (small $\lambda$) the new bound $B_{PRKZ}$ 
(\ref{shann_mixed}) dominates the results 
($B_B$) obtained by  Berta et al.\cite{BCRR10}; ($B_{KLJR}$) by Korzekwa et al.
\cite{KLJR14}; ($B_{KPP}$) by Kurzyk et al.~\cite{kurzyk2017conditional} and our
earlier bound ($B_{RPZ3}$) from \cite{RPZ14}.
} \label{fig:qubit-lambda-plot}
\end{figure}

\bigskip

Let us note in passing that since for $N=2$  one has $s_1 \ge 1/\sqrt{2}$,  we obtain
\begin{equation}
\lambda_1 (1 + s_1) + \lambda_2 (1 - s_1) \leq \lambda_1 + s_1,
\end{equation}
which gives us another majorization relation 
\begin{equation}
 W^{(\lambda)}  \prec \{s_1,1-s_1,\lambda_1,\lambda_2\} = W \oplus \lambda,
\end{equation}
where $W= \{s_1,1-s_1\} $ is a vector used to obtain in~\cite{RPZ14}
the direct sum majorization relation  for pure states. 

Thus for the case of the Shannon entropy % $S$
$H$ we get
\begin{equation}
H(p) + H(q) = \tilde{H}(p \oplus q) \geq \tilde{H}(W^{(\lambda)})
\geq \tilde{H}(W \oplus \lambda) = H(W) + H(\lambda),
\end{equation}
where $\tilde{H}$ is a Schur-concave function $\tilde{H}:x \mapsto -\sum x_i \log(x_i)$ (it is not the same as the entropy, because it also allows vector arguments which are not normalized).
Note that $\{1\} \oplus W=W^{(\{1,0\})}$, 
so we are in agreement with the bounds obtained in~\cite{RPZ14}
for the Shannon entropy in the case of pure states.

\subsection{Qutrit measurements}
In this section we illustrate our result in the case of two qutrit 
measurements. Let $\ket{\psi_{AB}}$ be a two-qutrit pure state with a given 
Schmidt vector $\lambda$, determining the spectrum of the partial trace $\rho$. 
Taking a fixed orthogonal matrix 
\begin{equation} \label{eqn:matrix-O3}
O_3 = \frac{1}{\sqrt{6}}
\begin{bmatrix}
 \sqrt{2} & \sqrt{2} & \sqrt{2} \\
 \sqrt{3} & 0 & -\sqrt{3} \\
 1 & -2 & 1 \\
\end{bmatrix},
\end{equation}
we consider lower bounds for conditional entropies $H(X|B) + H(Y|B) $, 
and compare our lower bound~\eqref{RSM}, the bound by Berta et al.~\cite{BCRR10}
and the bound $B_{RPZ_3}$ from~\cite{RPZ14}. Here we also omit the negative
part of the bound, presenting (see Fig.~\ref{fig:qutrit-fig}) contour plots for the bounds on the $2D$ simplex of
eigenvalues $\lambda_1+\lambda_2+\lambda_3=1$. 

\begin{figure}[ht]
\centering
\subfloat[Bound 
$B_{PRKZ}$]{\includegraphics[width=0.25\linewidth]{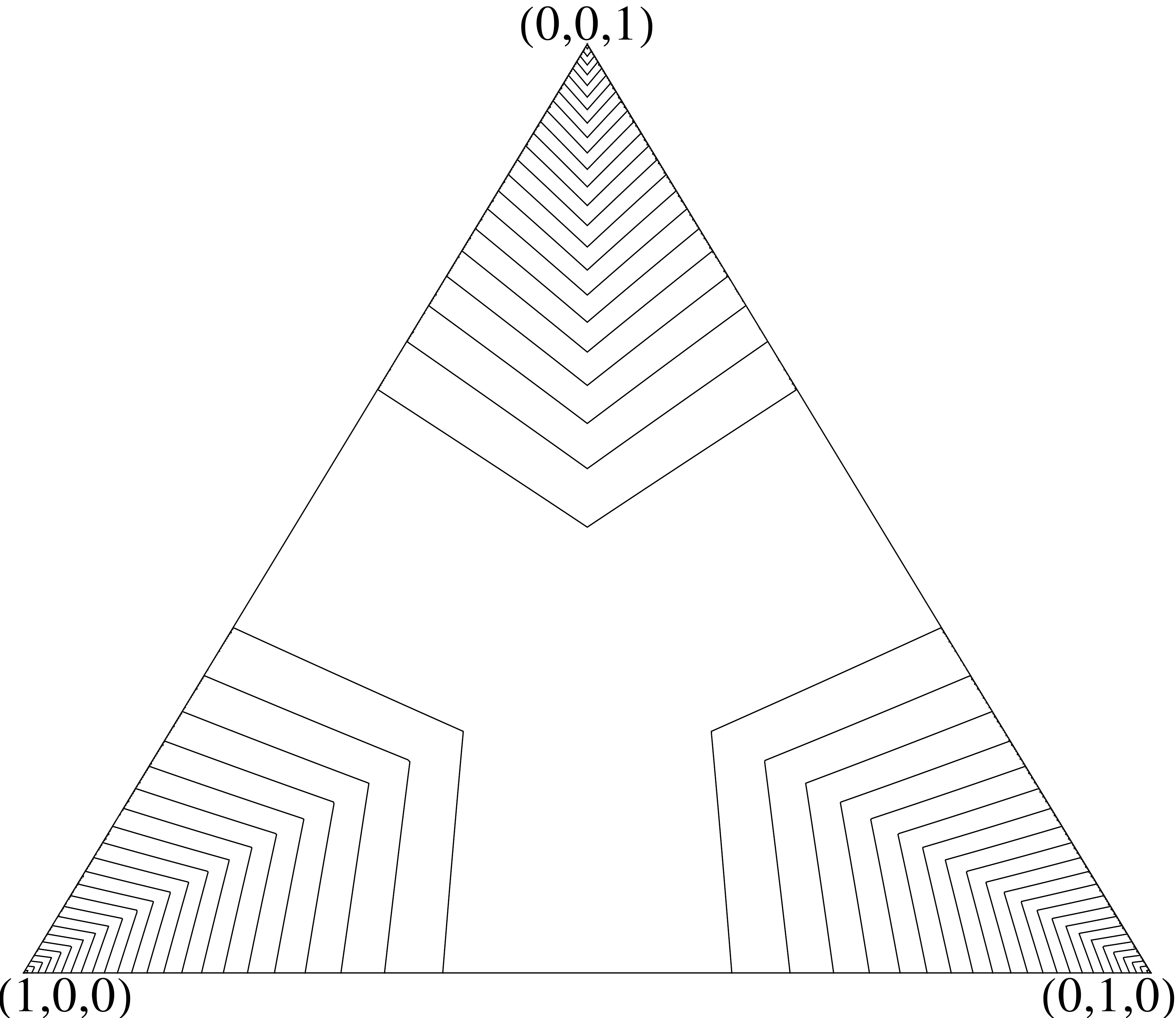}}
\ \ \
\subfloat[Bound 
$B_B$]{\includegraphics[width=0.25\linewidth]{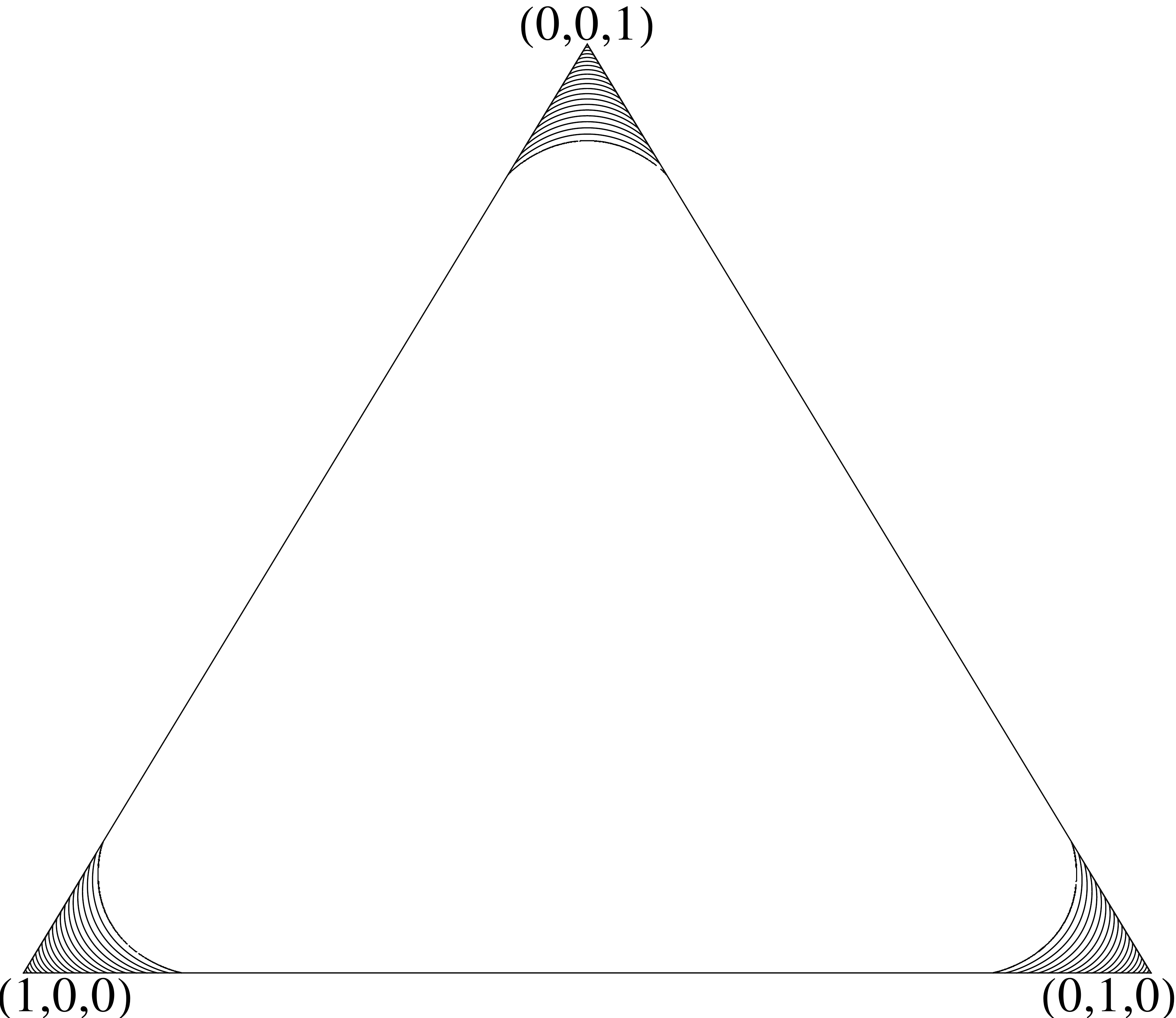}}
\ \ \
\subfloat[Bound 
$B_{RPZ_3}$]{\includegraphics[width=0.25\linewidth]{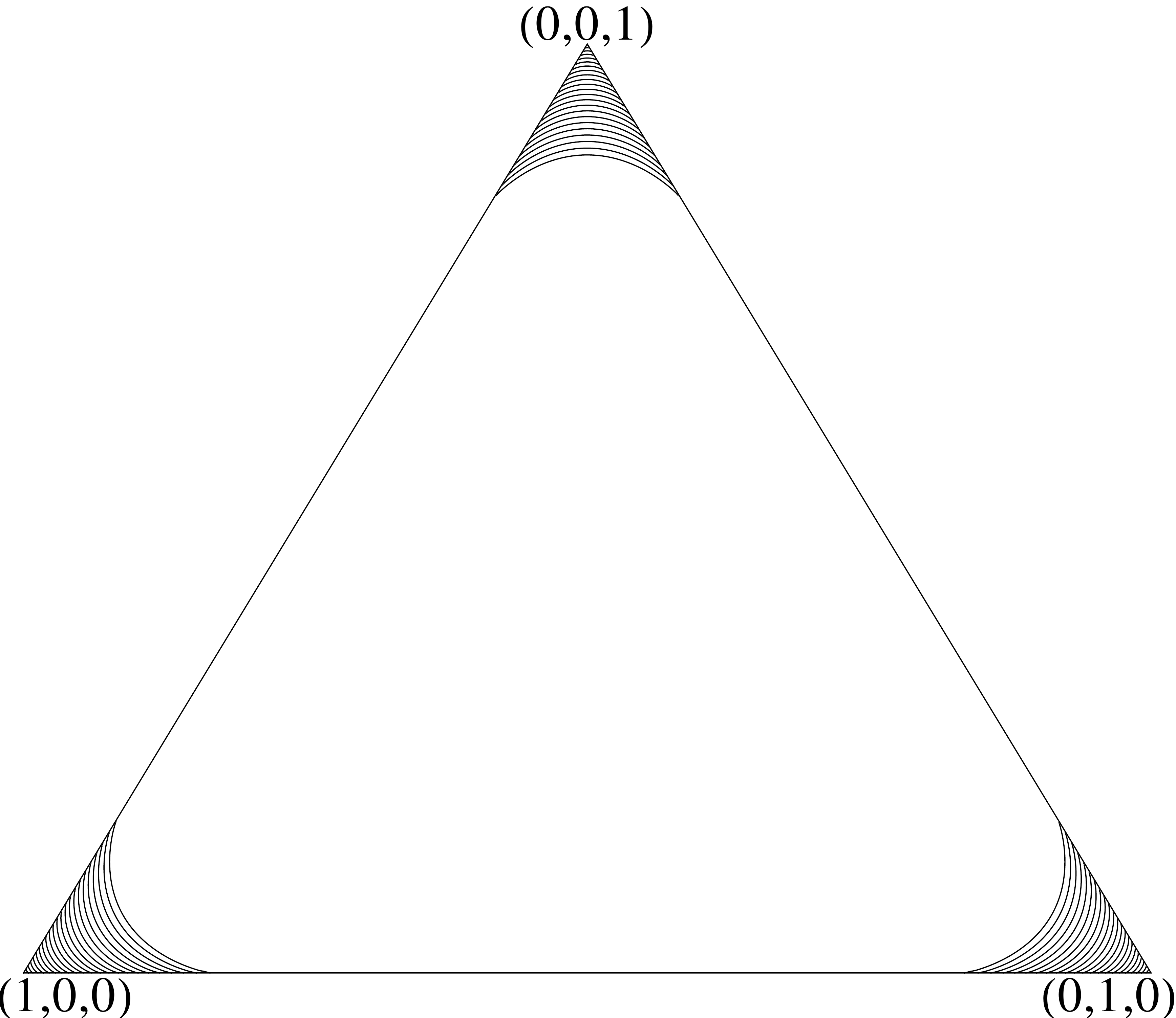}}
\caption{
Comparison of the lower bounds for conditional entropies for a chosen fixed set 
of two measurements determined by the orthogonal matrix $O_3$ introduced in 
Eq.~\eqref{eqn:matrix-O3} as a function of the ordered spectra $\lambda_1\ge 
\lambda_2 \ge \lambda_3$ represented for each bound by contour plots in the 
probability simplex obtained using the same heights of the countour lines. The new bound~\eqref{RSM} shown in panel (a) outperforms 
earlier results (b) bound $B_B$ from~\cite{BCRR10} and (c) bound $B_{RPZ_3}$ 
from~\cite{RPZ14}.
}
\label{fig:qutrit-fig}
\end{figure}

%%%%%%%%%%%%%%%%%%%%%%%%%%%%%%%%%%%%%%%%%%%%%%%%%%%%%%%%%%%%%%%%%%%%%%%%%%%%%%%%
\section{Concluding Remarks}
%%%%%%%%%%%%%%%%%%%%%%%%%%%%%%%%%%%%%%%%%%%%%%%%%%%%%%%%%%%%%%%%%%%%%%%%%%%%%%%%

The primary goal of our %(from the formal perspective algebraic)
investigations was to extend the quantum information toolbox relevant for
description of processes involving comparisons of probability vectors. The
major conceptual feature of our results is that they are based on the idea of
majorization, which in turn has several appealing practical interpretations due
to its ability to create a partial order and capture efforts necessary to
transform one probability distribution to the other. Apart from an expected
outcome of our work, namely the majorization UR for mixed states Eq.
\ref{major17}, we found an interesting uncertainty relation for conditional
entropies Eq. \ref{RSM}. Note that the results obtained in this work for the
setup of two orthogonal measurements performed on a mixed quantum state can also
be applied for a pure state of a bipartite system subjected to two measurements
performed locally on a single subsystem.

A natural open issue is to generalize the bounds presented here for the case of
three or more orthogonal measurements or to allow for a wider class of quantum
measurements. It might be also interesting to identify measurements for which
the bound provided are tight and to look for quantum states for which the
minimal values of the sum of entropies is achieved.

So far the uncertainty relations based on majorization have been established for
finite-dimensional coarse-grained continuous observables~\cite{LR2015}, and also
for quantum channels~\cite{RZ16}. Insightful and challenging
applications/extensions of the formalism and results presented in this paper are
thus situated in the field of continuous variables, where various majorization
techniques have already proven their usefulness \cite{Cerf1,Cerf2,Cerf3}.

\acknowledgements

It is a pleasure to thank  Patrick Coles for fruitful discussions. 
Financial support by the Polish National Science
Centre (NCN) under the grants number  2015/17/B/ST6/01872 (Z.P.), 
2016/22/E/ST6/00062 (A.K.) and DEC-2015/18/A/ST2/00274 (K.{\.Z}) 
% and by the John Templeton Foundation under the project No. 56033
is acknowledged.
\L .R. acknowledges
financial support by grant number 2014/13/D/ST2/01886 of the National
Science Center, Poland.

\appendix
%%%%%%%%%%%%%%%%%%%%%%%%%%%%%%%%%%%%%%%%%%%%%%%%%%%%%%%%%%%%%%%%%%%%%%%%%%%%%%%%
\section{Proof of lemma}
%%%%%%%%%%%%%%%%%%%%%%%%%%%%%%%%%%%%%%%%%%%%%%%%%%%%%%%%%%%%%%%%%%%%%%%%%%%%%%%%
In this appendix we restate and later proof 
Lemma~\ref{lemma:general-inequality}.

{\bf Lemma~\ref{lemma:general-inequality}.}\emph{Let $\lambda$ be a given 
vector of non-negative eigenvalues summing up to one.
Let $\rho$ be a mixed state with eigenvalues given by $\lambda$.
Let $p_i = \bra{i} \rho \ket{i}$ and $q_i =  \bra{a_i} \rho \ket{a_i}$. We 
assume, that $\{\ket{i}\}_i$ and $\{\ket{a_i}\}_i$ form orthonormal bases. 
Then we have the following inequality
\begin{equation}
p_1+ \dots + p_m + q_1+\dots + q_n \leq \lambda^{\downarrow} \cdot 
\mu^{\downarrow},
\end{equation}
where $\cdot$ denotes the scalar product. 
The vector $\mu$ of length $m+n$
is defined
\begin{equation}
\mu = \{1,1,\dots, 1\}  + (\sigma(A) \oplus (- \sigma(A)),
\end{equation}
for $A_{ij} = \scalar{a_i}{j}$ where
 $\sigma(A)$ denotes a vector of singular values of the
matrix $A$ and the symbol $\oplus$ represents concatenation
of two vectors. 
If necessary, we extend the vector $\sigma(A) \oplus (- \sigma(A))$
to length $m+n$ with zeros. 
}

Note, that the above Lemma is a generalization of Lemma 1
from~ \cite{PRZ13}, in the case of $\lambda = \{1,0,\dots,0\}$
we recover it.

%\begin{equation}
%p_1+ \dots + p_m + q_1+\dots + q_n \leq \lambda^{\downarrow} \cdot eig(M)^{\downarrow},
%\end{equation}
%where $\cdot$ denotes scalar product and

\begin{proof}
We define a matrix $M$ as
\begin{equation}
M = \ketbra{1}{1} + \dots + \ketbra{m}{m} + \ketbra{a_1}{a_1} + \dots + \ketbra{a_n}{a_n}.
\end{equation}
Next we calculate
\begin{equation}
\begin{split}
p_1+ \dots + p_m + q_1+\dots + q_n 
&=
\tr \rho (\ketbra{1}{1} + \dots + \ketbra{m}{m} + \ketbra{a_1}{a_1} + \dots + \ketbra{a_n}{a_n})\\
&=
\tr \rho M \leq  \lambda^{\downarrow} \cdot eig(M)^{\downarrow}.
\end{split}
\end{equation}
The last inequality is a von Neumann's trace inequality.
Now we note, that 
\begin{equation}
M = C^{\dagger} C,
\end{equation}
for a vector $C$ of length $m+n$ defined in \cite[lemma 1]{PRZ13}  
\begin{equation}
C = 
\begin{bmatrix}
\bra{1} \\ 
\bra{2} \\ 
\vdots \\
\bra{m} \\
\bra{a_1} \\ 
\bra{a_2} \\ 
\vdots \\
\bra{a_n} 
\end{bmatrix}.
\end{equation}
To calculate eigenvalues of $M$ in terms of the coefficients $A_{ij}$ we write
\begin{equation}
\begin{split}
eig(M) &= eig(C^{\dagger} C) = eig(C C^{\dagger}) \ \ \ 
%\text{ NOTE: this equality considers only non-zero eigenvalues} 
\\
&= 
eig
\begin{bmatrix}
\mathbb{I}_m & A^{\dagger} \\ 
A & \mathbb{I}_{n}
\end{bmatrix} 
= 
1+ eig
\begin{bmatrix}
0 & A^{\dagger} \\ 
A & 0
\end{bmatrix} 
\\ 
&= 
\{1 + \sigma(A)\} \oplus \{1 - \sigma(A)\} = \mu.
\end{split}
\end{equation} 
Note that the equality in the first line
of the above reasoning concerns only nonzero eigenvalues
of the matrices $C^{\dagger} C$ and $C C^{\dagger}$.
\end{proof}

\bibliographystyle{ieeetr}
\bibliography{mixmaj} 

\end{document}